\documentclass{article}

\usepackage{amsmath,amsthm}
\usepackage{algorithm}
\usepackage{algpseudocode}
\usepackage{graphicx}
\usepackage{nicefrac}
\usepackage{fullpage,authblk,enumitem,amssymb,xcolor}
\setlist{itemsep=0em,topsep=.5em,parsep=.2em} 
\usepackage[english]{babel}
\usepackage[hidelinks]{hyperref}
\usepackage[utf8]{inputenc}

\newcommand{\lebontau}{6\cdot (2k)^\ell}
\newcommand{\id}{\mathsf{id}}

\newcommand{\CONGEST}{\textsf{CONGEST}}

\newcommand{\F}{{\mathsf{out}}}
\renewcommand{\H}{{\mathsf{in}}}

\newcommand{\bp}{X}
\newcommand{\tp}{Y}

\newtheorem{theorem}{Theorem}
\newtheorem{lemma}{Lemma}

\newtheorem{definition}{Definition}

\newtheorem{fact}{Fact}

\title{Deterministic Even-Cycle Detection in Broadcast CONGEST
\thanks{Research supported in part by the 
European QuantERA project QOPT (ERA-NET Cofund 2022-25),
the French ANR projects DUCAT (ANR-20-CE48-0006),
the French PEPR integrated project EPiQ (ANR-22-PETQ-0007),
and the QuanTEdu-France project (22-CMAS-0001).}}

\author[1]{Pierre Fraigniaud}
\author[1]{Maël Luce}
\author[1]{Frédéric Magniez}
\author[2]{Ioan Todinca}
\affil[1]{
  {Université Paris Cité, CNRS, IRIF},
  {Paris},
  {France}
}
\affil[2]{
  {Université d'Orléans, INSA-Centre Val de Loire, LIFO},
 {Orléans},
  {France}   
}

\date{}
\begin{document}

\maketitle

\begin{abstract}
We show that, for every $k\geq 2$, $C_{2k}$-freeness can be decided in $O(n^{1-\nicefrac{1}{k}})$ rounds in the Broadcast \CONGEST{} model, by a \emph{deterministic} algorithm. This (deterministic) round-complexity is optimal for $k=2$ up to logarithmic factors  thanks to the lower bound for $C_4$-freeness by Drucker et al. [PODC 2014], which holds even for \emph{randomized} algorithms. Moreover it matches the round-complexity of the best known \emph{randomized} algorithms by Censor-Hillel et al. [DISC 2020] for $k\in\{3,4,5\}$, and by Fraigniaud et al. [PODC 2024] for $k\geq 6$. Our algorithm uses parallel BFS-explorations with deterministic selections of the set of paths that are forwarded at each round, in a way similar to what was done for the detection of odd-length cycles, by Korhonen and Rybicki [OPODIS 2017]. However, the key element in the design and analysis of our algorithm is a new combinatorial result bounding the ``local density'' of graphs without $2k$-cycles, which  we believe is interesting on its own.
\end{abstract}

\section{Introduction}\label{sec:intro}

In the context of distributed computing in networks, deciding $H$-freeness for a given (connected) graph~$H$ has attracted a lot of attention in the standard $\CONGEST$ model (see, e.g., the survey~\cite{Censor-Hillel21}). Indeed, this problem is inherently local, and thus the main concern is measuring the amount of information that must be exchanged between the nodes for solving it. Recall that the $\CONGEST$ model~\cite{Peleg2000} assumes $n$ processing nodes connected as an $n$-node graph~$G=(V,E)$. Computation proceeds as a sequence of rounds. During each round, every node can send an $O(\log n)$-bit message to each of its neighbors, receive the messages sent by its neighbors, and perform some individual computation. In the \emph{broadcast} version of the model, which is the one used in this paper, it is required that, at each round, each node sends the same message to all its neighbors. An algorithm decides $H$-freeness if, for every graph~$G$, the following holds: If $G$ contains a subgraph isomorphic to~$H$ then at least one node \emph{rejects}, otherwise  all nodes \emph{accept}. 

For every $k\geq 3$, let $C_k$ denote the $k$-node cycle. It is known that, for every $k\geq 2$, there exists a deterministic algorithm deciding $C_{2k+1}$-freeness in $O(n)$ rounds~\cite{KorhonenR17}, which is optimal up to logarithmic factors. It is however possible to decide the presence of even-size cycles in a sub-linear number of rounds. In particular, there exists a deterministic algorithm deciding $C_4$-freeness in $O(\sqrt{n})$ rounds~\cite{DruckerKO13}, which is optimal up to logarithmic factors, even for randomized algorithms. For every $k>2$, there exist randomized algorithms deciding $C_{2k}$-freeness in $O(n^{1-\nicefrac{1}{k}})$ rounds~\cite{Censor-HillelFG20,FraigniaudLMT24}. The algorithms in~\cite{Censor-HillelFG20}, based on the ``local threshold'' technique, apply to $2\leq k\leq 5$, whereas the algorithms in~\cite{FraigniaudLMT24}, based on the ``global threshold'' technique, apply to all $k\geq 2$. 

All aforementioned randomized algorithms are 1-sided, i.e., if $G$ contains a  $2k$-cycle, then the probability that at least one node rejects is at least~$\nicefrac23$, but if $G$ does not contain a  $2k$-cycle, then the probability that all nodes accept is~1. Of course, the error probability of these algorithms can be made as small as desired by mere repetition. Yet, it may still be the case that $G$ contains a $2k$-cycle that is not detected, even if this occurs with arbitrarily small probability. This raises the question of whether  $C_{2k}$-freeness can be decided by a \emph{deterministic} algorithm (which would thus provide absolute success) without increasing the round complexity. This is the case for all odd cycles of length~$\geq 5$, and for 4-cycles~\cite{DruckerKO13,KorhonenR17}. We show that this holds for all even cycles as well, by establishing the following result. 

\begin{theorem}\label{theo:c2k}
    For every $k\geq 2$, there is a deterministic algorithm solving $C_{2k}$-freeness in $O(n^{1-\nicefrac{1}{k}})$ rounds in the Broadcast $\CONGEST$ model. 
\end{theorem}

Our deterministic algorithm solving $C_{2k}$-freeness in $O(n^{1-\nicefrac{1}{k}})$ rounds is generic, parameterized by~$k$. For $k=2$, i.e., for $C_4$-freeness, our algorithm coincides with the one in~\cite{DruckerKO13}. In fact, our algorithm can be viewed as a generalisation of the latter. As for the algorithms in~\cite{Censor-HillelFG20,FraigniaudLMT24}, they distinguish the case of \emph{light} cycles, i.e., $2k$-cycles containing solely nodes with degree smaller than $n^{\nicefrac{1}{k}}$, from the case of \emph{heavy} cycles, i.e., $2k$-cycles containing at least one node with degree at least $n^{\nicefrac{1}{k}}$. Light cycles can be easily detected deterministically by flooding, in $O(n^{1-\nicefrac{1}{k}})$ rounds~\cite{Censor-HillelFG20,FraigniaudLMT24}, and the main issue is to show that heavy cycles can also be detected deterministically in $O(n^{1-\nicefrac{1}{k}})$ rounds. We show that this is indeed possible. 

Our algorithm (for detecting heavy cycles) relies on combining the Representative Lemma by Monien \cite{monien85}, with a new ``Density Theorem''. In a nutshell, the Representative Lemma can be used to show that, as already observed in~\cite{FraigniaudO19,KorhonenR17}, for each source node~$v$, it is not necessary for intermediate nodes to forward all prefixes of paths with endpoint~$v$ for eventually finding one path  forming a cycle containing~$v$, but forwarding a \emph{constant} number of prefixes suffices. Our density theorem is used to show that, if a node $v$ has to forward $\omega(n^{1-\nicefrac{1}{k}})$  prefixes of paths in total (i.e., corresponding to $\omega(n^{1-\nicefrac{1}{k}})$ source nodes~$v$), then there must exist a $2k$-cycle in the graph~$G$. More precisely, our theorem states the following.  

\begin{theorem}\label{lem:density}
    Let $G=(V,E)$ be an $n$-node graph. For every $v\in V$, and every integer $\ell\geq 1$, let $R_\ell(v)\subseteq V$ be the set of nodes that are reachable from $v$ by a simple path of length exactly~$\ell$. For every integer $k\geq 2$, if there exist $v\in V$ and $\ell\in\{1,\dots,k-1\}$ such that
    \[\sum_{u\in R_\ell(v)}\deg(u) > \lebontau \cdot n,\]
    then there is a $2k$-cycle in $G$.
\end{theorem}

Combining the Representative Lemma by Monien~\cite{monien85} with our density theorem, our algorithm for heavy $2k$-cycles boils down to flooding the network for $k$ steps with path-prefixes originated at all heavy nodes~$v$, under the following simple condition: at every step $\ell\in\{1,\dots,k\}$, if a node $u$ has to forward prefixes coming from more than $\lebontau \cdot n^{1-\nicefrac{1}{k}}$ heavy nodes then $u$ rejects.

If flooding is completed without any rejection then the Representative Lemma guarantees that a $2k$-cycle will be detected if any. Instead, if flooding aborts at some rejecting nodes, then the density theorem guarantees that this rejection is legitimate as there must exist a $2k$-cycle. 

\section{Model and Preliminary Results}\label{sec:prelim}

This section recalls the standard (Broadcast) \CONGEST\/ model, and takes benefit of the 4-cycle detection algorithm in~\cite{DruckerKO13} for introducing the reader with some of the techniques that will be reused throughout the paper. This is particularly the case of the Representative Lemma by Monien~\cite{monien85}, which is the basis for implementing a filtering technique enabling to bound the congestion of cycle-detection algorithms.

\subsection{The Broadcast CONGEST Model}

The \CONGEST\/ model~\cite{Peleg2000} assumes $n\geq 1$ processing nodes connected by a network modeled as an arbitrary $n$-node graph $G=(V,E)$. (All graphs are supposed to be connected and simple, i.e., no multiple edges nor self-loops.) Each node $v$ in $G$ has an identifier $\id(v)$ taken from a polynomial range of identifiers, and thus encoded on $O(\log n)$ bits. Each identifier is unique in the network. Computation proceeds as a sequence of synchronous rounds. All nodes start synchronously, at round~1. At each round, every node $v$ can (1)~send an $O(\log n)$-bit message to each of its neighbors in $G$, i.e., to each node $u\in N(v)$, (2)~receive the messages sent by its neighbors, and (3)~perform some individual computation. No limit is placed on the amount of computation that each node performs at each round. Initially, every node $v$ knows its identifier $\id(v)$ and the size $n$ of the graph it belongs to. No other information about the graph is provided to the nodes. All nodes perform the same algorithm, but the behavior of the nodes may vary along the course of execution of that algorithm, depending on the information acquired by them in each round (including their IDs). 

The Broadcast \CONGEST\/ model is a restricted variant of the \CONGEST\/ model which requires each node to send the \emph{same} $O(\log n)$ bit message to all its neighbors, at every round. (Of course, the messages sent by a same node $v$ at two different rounds may be different.)  

\subsection{Deciding $C_4$-Freeness}

The optimal (deterministic) algorithm for detecting 4-cycles in~\cite{DruckerKO13} can be artificially rewritten as Algorithm~\ref{alg:C4}. Let us say that a node $v$ is \emph{light} if its degree satisfies $\deg(v)\leq \sqrt{2n}$, and is \emph{heavy} otherwise. 

Phase~1 in Algorithm~\ref{alg:C4} is aiming at finding light 4-cycles, i.e., 4-cycles containing only light nodes. The for-loop of Instruction~\ref{for-loop-light-C4} consumes at most $\sqrt{2n}$ rounds, as it involves light nodes only, i.e., the light node $v$ has at most $\sqrt{2n}$ neighbors, and thus at most $\sqrt{2n}$ light neighbors. If a light 4-cycle is detected at Instruction~\ref{light-4-cycle}, then $v$ rejects appropriately. 

Phase~2 is aiming at finding heavy 4-cycles, that is, 4-cycles containing at least one heavy node. The for-loop of Instruction~\ref{for-loop-heavy-C4} also consumes at most $\sqrt{2n}$ rounds, merely because it is performed only if $|\mathsf{heavy}(v)|\leq  \sqrt{2n}$. The main observation is that it is legal for a node $v$ to reject if $|\mathsf{heavy}(v)|> \sqrt{2n}$. This is due to the following simple fact. 

\begin{lemma}\label{lem:densitu-for-4cycles}
    For every $n$-node graph $G=(V,E)$, if there exists $v\in V$ such that the inequality $\sum_{u\in N(v)}\deg (u)>2n$ holds, then $G$ contains a 4-cycle.
\end{lemma}
\begin{proof}
    Let $U=N(v)$ and $W=V\smallsetminus(N(v)\cup\{v\})$.

    If there is $y\in U\cup W$ that has two distinct neighbors $u$ and $u'$ in $U$, then $v,u,y,u'$ form a 4-cycle. Suppose on the contrary that every $y\in U\cup W$ has at most one neighbor in $U$. Then,
    \[
    \begin{split}
        \sum_{u\in U}\deg (u) &= \deg(v)+\sum_{u\in U} \deg_U(u)+\sum_{w\in W}\deg_U(w)\\
        &\leq |U|+|U|+|W|\\
        &\leq 2\times|U\cup W|\leq 2n.
    \end{split}
    \]
\end{proof}

\begin{algorithm}[tb]
\caption{Algorithm executed by every node $v\in V$ for deciding $C_4$-freeness in $G=(V,E)$}
\label{alg:C4}
\begin{algorithmic}[1]
    \State send $(\id(v),\deg(v))$ to all neighbors
    \State \textit{Phase 1: Looking for light 4-cycles}
    \If{$\deg(v)\leq \sqrt{2n}$} 
        \State $\mathsf{light}(v)\gets \{u\in N(v)\mid \deg(u)\leq \sqrt{2n}\}$
        \For{$u\in \mathsf{light}(v)$}\label{for-loop-light-C4}
            \State send $\id (u)$ to all neighbors; 
        \EndFor
        \If{$v$ has received $\id(w)\notin \{\id(v),\id(u),\id(u')\}$ from $u\neq u'$ both in $\mathsf{light}(v)$ } \label{light-4-cycle}
            \State \textbf{output}(reject), and terminate
        \EndIf 
    \EndIf
    \State \textit{Phase 2: Looking for heavy 4-cycles}
    \State $\mathsf{heavy}(v)\gets \{u\in N(v)\mid \deg(u)> \sqrt{2n}\}$
    \If{$|\mathsf{heavy}(v)|> \sqrt{2n}$} \label{inst:threshold-for-C4}
        \State \textbf{output}(reject) \label{inst:reject-legal-for-C4}
    \Else
        \For{$u\in \mathsf{heavy}(v)$}\label{for-loop-heavy-C4}
            \State send $\id (u)$ to all neighbors; 
        \EndFor
        \If{$v$ has received $\id(w)\notin \{\id(v),\id(u),\id(u')\}$ from $u\neq u'$ both in $N(v)$ } 
            \State \textbf{output}(reject)
        \Else 
            \State \textbf{output}(accept)
        \EndIf
    \EndIf
\end{algorithmic}
\end{algorithm}

Thanks to Lemma~\ref{lem:densitu-for-4cycles}, since $|\mathsf{heavy}(v)|> \sqrt{2n}$ implies that $\sum_{u\in N(v)}\deg (u)>2n$, we get that $G$~contains a 4-cycle, and thus it is indeed legal for $v$ to reject at Instruction~\ref{inst:reject-legal-for-C4}.

Our generic algorithm for detecting $2k$-cycles for every $k\geq 2$ follows the general idea of Algorithm~\ref{alg:C4} in the sense that:
\begin{enumerate}
    \item it is also split into two phases, one for light cycles (i.e., containing only nodes with degree smaller than $n^{\nicefrac{1}{k}}$), and one for heavy cycles (i.e., containing at least one node with degree at least $n^{\nicefrac{1}{k}}$), and
    \item it also utilizes a threshold as in Instruction~\ref{inst:threshold-for-C4}, which is tuned depending on~$k$. 
\end{enumerate}
In both phases, paths are broadcast among the nodes in the network. That is, every node $v$ participating in the broadcast sends its identifiers to its neighbors, which concatenate their identifiers, and forward the resulting 1-edge path to their neighbors. After $r$ rounds of such a process, every node $u$ receives a set of paths, each of the form $(\id(w_0),\dots,\id(w_{r-1}))$, concatenates its identifier to each of these paths, and forwards the resulting set of paths, each of the form $(\id(w_0),\dots,\id(w_{r-1}),\id(u))$ to all its neighbors. In itself, such a process would however create huge congestion. Indeed, the number of paths circulating in the network would grow exponentially. Nevertheless, reducing drastically the number of paths can  be achieved thanks to a simple application of the Representative Lemma by Monien~\cite{monien85}, as it was done in, e.g., \cite{FraigniaudO19,KorhonenR17}, which is explained a bit further in the text. Before that, let us quickly cover the study of so-called \emph{light} cycles, for which broadcast works without filtering. 

\subsection{Detection of Light $2k$-Cycles} 
\label{subsubsec:detection-light-cycles}

For every $k\geq 2$, the detection of \emph{light} $2k$-cycles, i.e. of $2k$-cycles containing only nodes of degree at most $n^{\nicefrac{1}{k}}$, can be done in a straightforward manner in $O(n^{1-\nicefrac{1}{k}})$ rounds. It is indeed sufficient to consider the subgraph $G_{light}$ of the input graph $G$ induced by the light nodes of~$G$. The detection proceeds by looking for all $2k$-cycles $G_{light}$ passing through~$v$, for all nodes $v\in V(G_{light})$ in parallel. 

Specifically, the algorithm proceeds by flooding, during $k$ phases. At the first phase (which lasts one round), every $v\in V(G_{light})$ forms the path $P=(v)$ consisting of the single node~$v$, and sends it to all its neighbors in $G_{light}$. At every phase~$p\geq 2$, for every node $v \in V(G_{light})$, and for every simple path $P=(u_1,\dots,u_{p-1})$ received by $v$ during phase $p-1$, if $v\notin \{u_1,\dots,u_{p-1}\}$, then $v$ appends its identifier to $P$ for forming the path $P'=(u_1,\dots,u_{p-1},v)$, which is forwarded to all of $v$'s neighbors in $G_{light}$. 

After $k$ phases, if a node $v$ has received a simple path $P=(u_1,\dots,u_k)$ from a neighbor~$u_k$, and a simple path $P'=(u'_1,\dots,u'_k)$ from a neighbor~$u'_k$, with $u_1=u'_1$, $P\cap P'=\{u_1\}$, and $v\notin P\cup P'$, then $v$ rejects. Otherwise $v$ accepts. We show that this simple flooding algorithm detects light $2k$-cycles in $O(n^{1-\nicefrac{1}{k}})$ rounds. 

\begin{lemma}\label{lem:complexity-light-cycles}
    For every $k\geq 2$, there is a deterministic algorithm running in $O(n^{1-\nicefrac{1}{k}})$ rounds in $n$-node graphs under the Broadcast $\CONGEST$ model  such that, for every graph $G$, if $G$ contains a \emph{light} $2k$-cycle (i.e., a cycle containing only nodes of degree smaller than $n^{\nicefrac{1}{k}}$) then at least one node rejects, otherwise all nodes accept.   
\end{lemma}

\begin{proof}
We analyze the simple flooding algorithm described above. Whenever a node rejects, this is because it has received a simple path $P=(u_1,\dots,u_k)$ from~$u_k$, and a simple path $P'=(u'_1,\dots,u'_k)$ from~$u'_k$, with $u_1=u'_1$, $P\cap P'=\{u_1\}$, and $v\notin P\cup P'$, which implies that $(u_1,\dots,u_k,v,u'_k,\dots,u'_2)$ is a $2k$-cycle, i.e., there is indeed a light $2k$-cycle in~$G$. Conversely, let $(w_1,\dots,w_{2k})$ be a light $2k$-cycle. The two paths $P=(w_1,\dots,w_k)$ and $P'=(w_1,w_{2k},\dots,w_{k+1})$ will be received by node $w_{k+1}$ after $k$ phases, leading node $w_{k+1}$ to reject, as desired. 

For the round-complexity, it is sufficient to notice that, by definition, $G_{light}$ has maximum degree at most $n^{\nicefrac{1}{k}}$. It follows that, for every $v\in V(G_{light})$, and for every $p\in\{0,\dots,k-1\}$, the number of simple paths of length $p$ with one end equal to~$v$ is at most  $n^{\nicefrac{p}{k}}$. As a consequence, the round complexity of the flooding algorithm is 
\[
\sum_{p=0}^{k-1} (p+1)\cdot n^{\nicefrac{p}{k}} \leq k^2 \cdot n^{(k-1)/k} = k^2 \cdot  n^{1-\nicefrac{1}{k}}=O(n^{1-\nicefrac{1}{k}}),
\]
as claimed. 
\end{proof}

The main difficulty is detecting heavy cycles, that is, cycles containing at least one node of degree at least $n^{\nicefrac{1}{k}}$. Among other techniques, one needs \emph{filtered flooding}, explained in the next section.  

\subsection{Filtered Flooding}

\subsubsection{Representative Lemma} 

Monien~\cite{monien85} defines a \emph{representative} of a family of subsets of ground set $[n]=\{1,\dots,n\}$ as follows. 

\begin{definition}\label{def:q-rep}
For every integer $n\geq 1$, every family $\mathcal{A}\subseteq 2^{[n]}$ of subsets of $[n]$, and every $q\in [n]$, a family of sets $\mathcal{B}\subseteq \mathcal{A}$ is a \emph{$q$-representative} of $\mathcal{A}$ if, for every set $X\subseteq [n]$ of size $|X|\leq q$, the following holds:
\[
\exists A\in\mathcal{A} \,\text{ such that } A\cap X = \varnothing 
\;\iff\; 
\exists B \in \mathcal{B} \,\text{ such that } B \cap X = \varnothing.
\]
\end{definition}

Note that it follows directly from the definition that the $q$-representativity property is transitive, which is important as our algorithm for $C_{2k}$-freeness will perform several nested iterations of computing a $q$-representative set.  

\begin{fact}\label{lem:rep_transi}
For every $n\geq 1$, $\mathcal{C}\subseteq\mathcal{B}\subseteq\mathcal{A}\subseteq 2^{[n]}$, and $q\in [n]$, if $\mathcal{B}$ is $q$-representative for $\mathcal{A}$, and $\mathcal{C}$ is $q$-representative for $\mathcal{B}$, then $\mathcal{C}$ is $q$-representative for $\mathcal{A}$.
\end{fact}

The following lemma says that if the sets in the family $\mathcal{A}$ have bounded size~$p$, then there exists a $q$-representative of $\mathcal{A}$ with bounded size. 

\begin{lemma}[Monien~\cite{monien85}]\label{lem:erdos}
For every integer $n\geq 1$, every $(p,q)\in [n]\times [n]$ such that $p+q\leq n$, and every family $\mathcal{A}\subseteq 2^{[n]}$ of subsets of $[n]$, if $|A|\leq p$ for all $A\in\mathcal{A}$, then  
there exists a $q$-representative family $\mathcal{B}\subseteq \mathcal{A}$ of $\mathcal{A}$ such that 
\[
|\mathcal{B}|\leq \binom{p+q}{p}.
\]
\end{lemma}

To get a flavor of this lemma, observe that if $p+q= n$ and $\mathcal{A}=\{A\subseteq [n], |A|=p\}$, then the bound of Lemma~\ref{lem:erdos} is tight. Indeed, $\mathcal{A}$ is then the set of all subsets of $[n]$ of size $p$, meaning that $|\mathcal{A}|=\binom{n}{p}$. Now suppose that $\mathcal{B}\subsetneq\mathcal{A}$, then any $A\in\mathcal{A}\smallsetminus\mathcal{B}$ is such that $A\cap\big([n]\smallsetminus A\big)=\varnothing$ but there is no $B\in\mathcal{B}$ that does not intersect $[n]\smallsetminus A$. The only $q$-representative family of $\mathcal{A}$ is itself.

\subsubsection{Application to Cycle Detection} 
\label{subsubsec:application-filtering}

The Representative Lemma (Lemma~\ref{lem:erdos}) is particularly useful in the context of $2k$-cycle detection for limiting congestion. Indeed, let us assume that one is questioning whether there is a $2k$-cycle in the $n$-node graph $G=(V,E)$ containing a given node~$v\in V$. One way to proceed is to let~$v$ broadcast its identifier for $k$ rounds. That is, at step~$p=0$, $v$ sends $\id(v)$ to all its neighbors. Then, at step~$p=1$, every neighbor $u\in N(v)$ appends its identifier to encode the 1-edge path $(v,u)$, and forwards this path to all its neighbors. More generally, if the flooding process is not filtered, then, at step~$p\in\{1,\dots,k-1\}$, a node $u$ receiving a simple path $P=(v_0,\dots,v_{p-1})$ with $v_0=v$ appends $u$ to $P$ whenever $u\notin \{v_0,\dots,v_{p-1}\}$, and forwards the resulting augmented path to all its neighbors. At the end of step~$k-1$, if a node $u$ has received two simple paths $P=(v_0,\dots,v_{k-1})$ and $P'=(v'_0,\dots,v'_{k-1})$ with $v_0=v'_0=v$, $u\notin P\cup P'$, and $P\cap P'=\{v\}$, then $u$ has detected a $2k$-cycle containing~$v$. 

The absence  of filtering in the above process may result in an exponential increase of the number of paths to be forwarded by an intermediate node~$u$. This can be avoided using Lemma~\ref{lem:erdos} by the following filtering process. Let us assume that, at step~$p\in\{1,\dots,k-1\}$, a node $u$ receives a collection $\mathcal{A}$ of $p$-node simple paths, all with endpoint~$v$. Each path can be viewed as a set of $p$ nodes, i.e., a subset of the $n$-node set $V$ with cardinality~$p$. 
Let $q=2k-p$. Lemma~\ref{lem:erdos} says that there exists $\mathcal{B}\subseteq \mathcal{A}$ of cardinality $|\mathcal{B}|\leq \binom{2k}{p}$ such that, for every simple path 
$X=(u_0,\dots,u_{2k-p-1})$ from $u_0=u$ to $u_{2k-p-1}\in N(v)$, if there exists a path $A\in \mathcal{A}$ that does not intersect with~$X$, i.e., such that $A\cup X$ is a $2k$-cycle containing~$v$, then there exists a path $B\in \mathcal{B}$ that does not intersect with~$X$, i.e., such that $B\cup X$ is also a $2k$-cycle containing~$v$. The filtering process consists for node $u$ to forward the family $\mathcal{B}$, after concatenating itself to every path in it, instead of $\mathcal{A}$. By the filtering technique, we have the following. 

\begin{fact}\label{fact:not-many-paths}
    Each intermediate node $u$ forwards at most $\binom{2k}{p}$ paths of size $p+1$ and of endpoint~$v$ at each round~$p\in\{0,\dots,k-1\}$.
\end{fact}

Hence, the number of paths forwarded by a node $u$ is constant for a fixed~$k$. As a consequence, we get the following. 

\begin{fact}\label{fact:cycle-through-one-vertex}
    For every $k\geq 2$, every $n$-node graph $G=(V,E)$, and every $v\in V$, checking whether there is a $2k$-path in $G$ containing~$v$ takes at most $\sum_{p=0}^{k-1}(p+1)\binom{2k}{p}<k2^{2k}$ rounds in the Broadcast \CONGEST\/ model.
\end{fact}

\section{Deterministic Even-Cycle Detection}\label{sec:cycle_detect}

This section is dedicated to the proof of Theorem~\ref{theo:c2k} assuming that the density theorem holds (cf. Theorem~\ref{lem:density}). The density theorem will be established in the next section. We start here by describing our algorithm for deciding $C_{2k}$-freeness, and then we proceed to the proof of Theorem~\ref{theo:c2k}. 

\subsection{Algorithm Description}

Algorithm~\ref{alg_C2k} solves the detection of $2k$-cycles.  As Algorithm~\ref{alg:C4}, it is split into two phases, one aiming at detecting light $2k$-cycles, and one aiming at detecting heavy $2k$-cycles, where a $2k$-cycle is light if it contains only nodes of degree smaller than $n^{\nicefrac{1}{k}}$, and it is heavy otherwise. The second phase, that is, the detection of heavy cycles, uses the filtering techniques (see Instruction~\ref{ins:set_P}) based on Lemma~\ref{lem:erdos}, and detailed in Section~\ref{subsubsec:application-filtering}. The detection of light cycles has already been described and analyzed in Section~\ref{subsubsec:detection-light-cycles}, and this section focuses on Phase~2, i.e., the detection of heavy cycles, starting at Instruction~\ref{ins:start-heavy-here}. 

\begin{algorithm}[p]
\caption{Algorithm run by every node $v\in V$ for deciding $C_{2k}$-freeness in $G=(V,E)$}\label{alg_C2k}
\begin{algorithmic}[1]
    \State \textsf{send} $(\id(v),\deg(v))$ to all neighbors \label{ins:basic-first-step}
    \State \textit{Phase 1: Looking for light $2k$-cycles}
    \If{$\deg(v)< n^{\nicefrac{1}{k}}$}\label{ins:light_begin}
        \State \textsf{send} $P=(\id(v))$ to all neighbors \label{ins:light_BFS_init}
        \For{$i=1$ to $k-1$}
            \State \textsf{receive} paths sent by neighbors
            \State let $\mathcal{P}$ be the set of received paths (not containing $v$)\label{ins:apply-filtering-light}
            \For{$P\in\mathcal{P}$}
                \State $P\gets (P,v)$ \label{ins:light_BFS}
                \State \textsf{send} $P$ to neighbors
            \EndFor
        \EndFor
        \State  \textsf{receive} paths sent by neighbors, and let $\mathcal{P}$ be the received set of paths
        \If{$\exists P,P'\in \mathcal{P}$ of same origin $u$ such that $v\notin P\cup P'    \text{ and } P\cap P'=\{u\}$} 
            \State \textsf{output}(reject) and terminate
        \EndIf 
    \EndIf\label{ins:light_end}
    \State \textit{Phase 2: Looking for heavy $2k$-cycles} \label{ins:start-heavy-here}
    \State $\mathsf{heavy}(v)\gets \{u\in N(v) \mid \deg(u)\geq n^{\nicefrac{1}{k}}\}$
    \For{$w\in V$}
        \State \textbf{if} $w\in \mathsf{heavy}(v)$ \textbf{then} $\mathcal{Q}(w,v)\gets \{(\id(w),\id(v))\}$ \textbf{else} $\mathcal{Q}(w,v)\gets \varnothing$ \label{ins:heavy_init}
    \EndFor
        \For{$\ell=1$ to $k-1$}\label{the-for-loop-for-heavy-C2k}
            \State $W(v)\gets\{w\in V\mid \mathcal{Q}(w,v)\neq\varnothing\}$; \label{ins:set_W}
            \If{$|W(v)|> \lebontau \cdot n^{1-\nicefrac{1}{k}}$}\label{ins:heavy_thres}
                \State \textsf{output}(reject)  and terminate
            \Else
                \For{$w\in W(v)$}
                    \State $\mathcal{P}(w,v)\gets$ filtering applied to $\mathcal{Q}(w,v)$ \label{ins:set_P}
                    \For{$P\in \mathcal{P}(w,v)$}
                        \State send $P$ to neighbors \label{ins:heavy_broadcast}
                    \EndFor
                \EndFor
                \For{$w\in V$}
                    \State $\mathcal{Q}(w,v)\gets \{(P,v)\mid P\in \mathcal{P}(w,u) \land u\in N(v) \land v\notin P\}$ \label{ins:set_Q}
                \EndFor
            \EndIf
        \EndFor
        \If{$\exists w\in V, \exists P,P'\in \mathcal{Q}(w,v) \mid P\cup P'=C_{2k}$}\label{ins:heavy_detect}
            \State \textsf{output}(reject)
        \Else
            \State \textsf{output}(accept)
        \EndIf\label{ins:heavy_end}
    \end{algorithmic}
\end{algorithm}

In Algorithm~\ref{alg_C2k}, every node $v$ of a graph $G=(V,E)$ maintains a collection of local variables. The set $\mathsf{heavy}(v)$ contains all heavy neighbors of~$v$ in~$G$, thanks to Instruction~\ref{ins:basic-first-step}. For every $w\in V$, the set $\mathcal{Q}(w,v)$ contains a collection of simple paths with endpoints $v$ and~$w$. At the beginning of $\ell$-th iteration of the for-loop of Instruction~\ref{the-for-loop-for-heavy-C2k}, $\mathcal{Q}(w,v)$ contains paths of length exactly~$\ell$, which are eventually updated at the end of the $\ell$-th iteration (cf. Instruction~\ref{ins:set_Q}) to paths of length~$\ell+1$. At each iteration, the set $W(v)$ is the set of nodes~$w$ such that there is at least one path from $w$ to $v$ in $\mathcal{Q}(w,v)$, i.e., $\mathcal{Q}(w,v)$ is not empty. 

The main point in Algorithm~\ref{alg_C2k} is the test performed at Instruction~\ref{ins:heavy_thres}, which stipulates that if $W(v)$ is too big, i.e., if $v$ is connected to too many (heavy) nodes $w$ by a path of length~$\ell$ at iteration~$\ell$, then $v$ rejects. If $v$ does not reject, then it carries on the flooding of path-prefixes, by applying filtering for preserving the fact that, for each $w\in W(v)$, $|\mathcal{P}(w,v)|\leq \binom{2k}{\ell+1}$ at every iteration (cf. Fact~\ref{fact:not-many-paths}). At the end of each iteration~$\ell$ of the for-loop of Instruction~\ref{the-for-loop-for-heavy-C2k}, node~$v$ appends $\id(v)$ to each path received during that iteration (see Instruction~\ref{ins:set_Q}). That is, node $v$ appends $\id(v)$ to each path $P\in \mathcal{P}(w,u)$ not containing~$v$, for all neighbors $u\in N(v)$, and it resets $\mathcal{Q}(w,v)$ accordingly. 

Finally, if node $v$ has received two paths $P$ and $P'$ of length~$k$, both in $\mathcal{Q}(w,v)$ for some $w\in V$, i.e., both with endpoints $v$ and $w$, such that the concatenation of $P$ and $P'$ forms a $2k$-cycle, then $v$ rejects. 

\subsection{Proof of Theorem~\ref{theo:c2k}}

In absence of the threshold condition at Instruction~\ref{ins:heavy_thres}, Algorithm~\ref{alg_C2k} merely consists of building longer and longer paths between $v$ and some nodes $w\in V$, such that if there exists $v$ and $w$ for which there exists two paths $P$ and $P'$ of length~$k$ between $v$ and $w$ that are internally disjoint, that is if $P\cup P'$ form a $2k$-cycle, $v$ will detect that fact, and reject accordingly. As already discussed before, the filtering of Instruction~\ref{ins:set_P} does not prevent the algorithm from finding such paths, if any. The main issue is that Algorithm~\ref{alg_C2k}  stops at iteration~$\ell$ whenever $|W(v)|> \lebontau \cdot n^{1-\nicefrac{1}{k}}$. Let us show that stopping under this condition is fine, as it implies the existence of a $2k$-cycle.

Let us assume that there exists a node $v$ such that, at iteration $\ell\in\{1,\dots,k-1\}$ of the for-loop of Instruction~\ref{the-for-loop-for-heavy-C2k}, $|W(v)|> \lebontau \cdot n^{1-\nicefrac{1}{k}}$. At iteration~$\ell$, $W(v)$ is the set of \emph{heavy} nodes $w$ such that there is a simple path of length exactly $\ell$ starting at $w$ and ending at~$v$. 
Using the notation of Theorem~\ref{lem:density}, let $R_\ell(v)\subseteq V$ be the set of nodes that are reachable from $v$ by a simple path of length exactly~$\ell$. We have $W(v)\subseteq R_\ell(v)$. It follows that 
\begin{equation}\label{eq:correctness}
\sum_{w\in R_\ell(v)}\deg(w) 
\geq \sum_{w\in W(v)}\deg(w) 
\geq |W(v)|\cdot n^{\nicefrac{1}{k}} 
> \lebontau \cdot n^{1-\nicefrac{1}{k}} \cdot n^{\nicefrac{1}{k}}
= \lebontau \cdot n.
\end{equation}
By Theorem~\ref{lem:density}, $G$ contains a $2k$-cycle, and thus node $v$ rejects rightfully.

It remains to show that the round complexity of Algorithm~\ref{alg_C2k} is $O(n^{1-\nicefrac{1}{k}})$. Phase~1, dedicated to the search of light $2k$-cycles, takes this many rounds, as established in Lemma~\ref{lem:complexity-light-cycles}. Let us show that Phase~2, dedicated to the search of heavy $2k$-cycles, has the same complexity. By Fact~\ref{fact:cycle-through-one-vertex}, for every node $v$ and every heavy node~$w$, the final set $\mathcal{Q}(w,v)$ at iteration $\ell=k-1$ at Instruction~\ref{ins:set_Q} is built after at most $k\cdot 2^{2k}$ rounds. By the threshold condition of Instruction~\ref{ins:heavy_thres}, the number of families $\mathcal{P}(w,v)$ to be transmitted by~$v$ is at most $\lebontau\cdot n^{1-\nicefrac{1}{k}}$. So, in total, Phase~2 of  Algorithm~\ref{alg_C2k}  performs in at most $k \cdot 2^{2k} \cdot \lebontau \cdot n^{1-\nicefrac{1}{k}}$ rounds, that is $O(n^{1-\nicefrac{1}{k}})$ rounds for a fixed~$k$. This completes the proof of Theorem~\ref{theo:c2k} (under the assumption that Theorem~\ref{lem:density}
 holds). \qed

\section{Proof of Density Theorem}

\subsection{General Construction}
This section is dedicated to the proof of Theorem~\ref{lem:density}.
Let $G=(V,E)$ be an $n$-node graph. Let $k\geq 2$, $\ell\in\{1,\dots,k-1\}$, and $v\in V$. Let $R(v)$ be the set of nodes of $G$ that are reachable from $v$ by a simple path of length exactly~$\ell$, and let us assume that 
\[
\sum_{w\in R(v)}\deg(w)>\lebontau\cdot n.
\]
Our goal is to show  that there is a $2k$-cycle in~$G$. In the following, we fix 
\begin{equation}\label{eq:tau}
\tau = \lebontau.
\end{equation}

Let $F$ be the set of edges incident to at least one node in~$R(v)$. Let $F_{int}\subseteq F$ be the set of edges whose both endpoints are in~$R(v)$, and let $F_{ext}=F\smallsetminus F_{int}$.   

\begin{lemma}[S. Burr \cite{burr1982}]\label{lem:de-Burr}
Every $m$-edge graph contains a bipartite subgraph with at least $m/2$ edges.
\end{lemma}

Thanks to Lemma~\ref{lem:de-Burr}, there exists a bipartite subgraph of the graph $G[F_{int}]$ induced by $F_{int}$ with at least $|F_{int}|/2$ edges. Let $F'_{int}\subseteq F_{int}$ be the set of edges of this bipartite graph, hence $|F'_{int}|\geq |F_{int}|/2$.

Furthermore, $F_{ext}$ also induces a bipartite subgraph as all of its edges have exactly one end in $R(v)$.

Let $H$ be the bipartite graph defined as the subgraph of $G$ induced by $F'_{int}$ if $|F'_{int}|\geq |F_{ext}|$, or as the subgraph of $G$ induced by $F_{ext}$ otherwise. That is, 
\[
H=\left\{\begin{array}{ll}
G[F'_{int}] & \mbox{if $|F'_{int}|\geq |F_{ext}|$} \\
G[F_{ext}] & \mbox{otherwise.}
\end{array}\right.
\]
Let $(\bp,\tp)$ be a partition of the vertices of $H$, such that \[\bp\subseteq R(v).\] Note that some nodes in $\tp$ may also belong to $R(v)$. $H$ satisfies the following.

\begin{fact}\label{fact:size_edges}
$H=(\bp,\tp,E_H)$ is a bipartite subgraph of $G$ with at least $\frac16\, \tau\, n$ edges.
\end{fact}

\begin{proof}
We have 
\[
|F|= |F_{ext}|+|F_{int}| \leq |F_{ext}|+2|F'_{int}|\leq 3|E_H|.
\]
Since
\[
2 |F| \geq \sum_{w\in R(v)}\deg(w) >\tau \cdot n,
\]
the claim follows.
\end{proof}


To prove the existence of a $2k$-cycle in $G$, we will prove that there exist three simple paths $P$, $P'$, and $P''$ in $G$ such that:
\begin{itemize}
    \item $P$ is of length $\lambda$ for some  $\lambda\in\{1,\dots,\ell\}$ connecting a node $x\in \bp$ to some node $u\in V$.
    \item $P'$ is a path in $H$ of length $2k-2\lambda-1$ connecting $x$  to a node $y\in \tp$. Note that, since $H$ is bipartite, $P'$~alternates between nodes in $\bp$ and nodes in~$\tp$. 
    \item $P''$ is a path of length $\lambda+1$ connecting $y$ to~$u$.
\end{itemize}
Moreover, our construction will guarantee that $P$, $P'$ and $P''$ are internally disjoint, in the sense that $P\cap P'=\{x\}$, $P\cap P''=\{u\}$, and $P'\cap P''=\{y\}$. This is sufficient for establishing Theorem~\ref{lem:density} as $P\cup P'\cup P''$ is a $2k$-cycle in $G$. To exhibit these three paths, let us introduce some notations.

\subsection{The Sets In and Out.}\label{subsec:set_in_out}

For every $u\in V$, and every $i\in\{0,\dots,\ell\}$, let us denote  by $\mathcal{Q}_i(u)$ the set of simple paths of length exactly~$i$ with one endpoint equal to~$u$, and the other endpoint equal to some node in~$\bp$. 

For every $u\in V$, and $i\in\{0,\dots,\ell\}$, we define the two sets $\H_i(u)$ and $\F_i(u)$ as subsets of edges from $H$. Intuitively, the set $\F_i(u)$ can be viewed as a set of edges that $u$ sends to its neighbors at round~$i$ in a (virtual) distributed protocol that broadcasts sets of edges of $H$ throughout the network~$G$, and $\H_i(u)$ can be viewed as a set of edges that $u$ receives from its neighbors at round~$i$ of this protocol. 
Let $E_H(u)$ denote the set of edges incident to $u$ in~$H$. For every $u\in V$, let 
\begin{equation}\label{eq:OUT_init}
    \F_0(u) = \left\{\begin{array}{cl}
    E_H(u) & \mbox{if $u\in \bp$,} \\
    \varnothing & \mbox{otherwise.}
    \end{array}\right.
\end{equation}
That is, $\F_0(u)$ can be viewed as the initialization of the aforementioned (virtual) broadcast protocol, i.e., initially, every node in~$X$ sends its incident edges in $H$ to its neighbors in~$G$. Now, let us define the sets $\H_i(u)$ and $\F_i(u)$ for every $u\in V$ and $i\geq 1$, where, again, $\H_i(u)$ can be viewed as the set of edges received by~$u$ at round~$i$ (which were sent by $u$'s neighbors at round $i-1$), and $\F_i(u)$ can be viewed as the set of edges forwarded by~$u$ at round~$i$. A key point is that $u$ does not forward all the received edges, but a carefully chosen subset of these edges.

Formally, for every $u\in V$ and every $i\in\{1,\dots,\ell\}$, let
\begin{equation}\label{eq:IN_recur}
    \H_i(u) = \bigcup_{\{w\in N(u) \;\mid\; \exists P\in\mathcal{Q}_{i-1}(w),\; (P,u)\in\mathcal{Q}_i(u)\}} \F_{i-1}(w).
\end{equation}
That is, $\H_i(u)$ merges all edges from sets $\F_{i-1}(w)$ for all neighbors $w$ of $u$ such that at least one path in $\mathcal{Q}_{i-1}(w)$ can be extended into a path in~$\mathcal{Q}_i(u)$ (see Figure~\ref{fig:in-set}).

\begin{figure}
    \centerline{\includegraphics[scale=0.5]{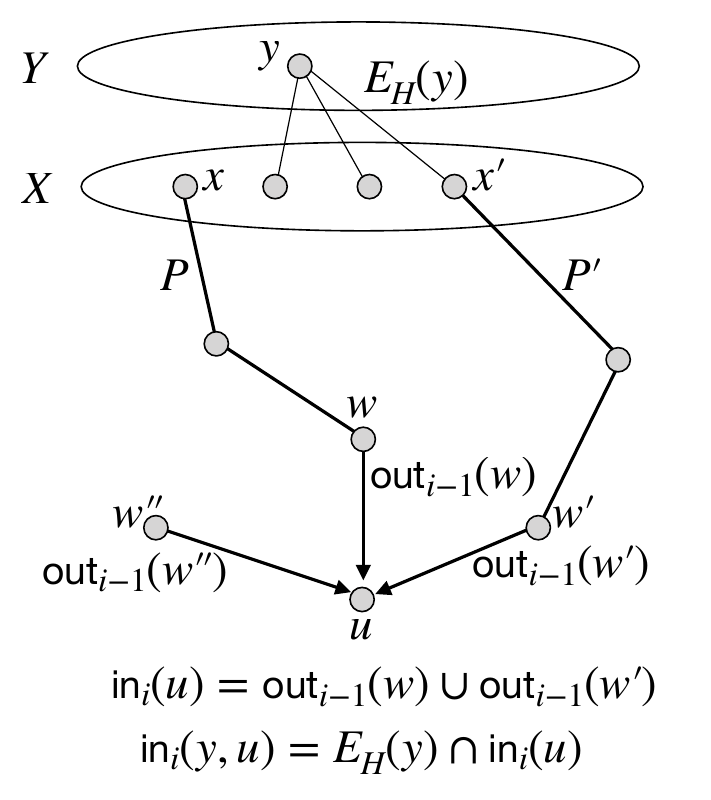}}
    \caption{\sl Construction of the set $\H_i(u)$ from the sets $\F_{i-1}(w)$, $w\in N(u)$. In the figure, $\H_i(u)=\F_{i-1}(w) \cup \F_{i-1}(w')$ because there is a simple path $P$ (resp., $P'$) of length $i-1$ from $w$ (resp., $w'$) to $X$ that can be extended to a simple path of length $i$ from $u$ to $X$. For every $y\in Y$, $\H_i(y,u)=E_H(y)\cap\H_i(u)$. }
    \label{fig:in-set}
\end{figure}

To define $\F_i(u)$, we first define the sets $\H_i(w,u)$ and $\F_i(w,u)$, the subsets of edges respectively in $\H_i(u)$ and $\F_i(u)$ incident to node $w\in  \bp\cup \tp$ (see Figure~\ref{fig:out-set}). For every $i\in\{1,\dots,\ell\}$, for every node $w$ of~$H$, i.e., for every $w\in \bp\cup \tp$, and every $u\in V$, let
\begin{equation}\label{eq:inx-equal-iny}
\H_i(w,u)=E_H(w)\cap\H_i(u). 
\end{equation}
For every vertex $u\in V$, and every $i \in  \{1,\dots,\ell\}$, we construct the edge-set $\F_i(u)$ by defining, for each $y\in \tp$, a subset $\F_i(y,u)$ of $\F_i(u)$ containing edges of $\F_i(u)$ incident to~$y$.
First, for every $y \in \tp$ and $u\in V$, we set 
\[
\F_0(y,u)=E_H(y)\cap\F_0(u)=\left\{\begin{array}{cl}
    \{y,u\} & \mbox{if $\{y,u\}\in E_H$} \\
    \varnothing & \mbox{otherwise}
    \end{array}\right.
\]
For every $u\in V$ and $i\in \{1,\dots,\ell\}$, and for every $y \in \tp$, we set
\begin{equation}\label{eq:OUT_recur}
    |\H_i(y,u)|<(2k)^i \Longrightarrow \F_i(y,u)= \H_i(y,u).
\end{equation}
The definition of $\F_i(y,u)$ when  $|\H_i(y,u)|\geq (2k)^i$ requires more care. 
Let $H_i^{(u)}$ be the subgraph of $H$ induced by all edges in $\cup_{y\in \tp}\H_i(y,u)$ where one keeps only the large sets $\H_i(y,u)$, that is, 
\begin{equation}\label{eq:H_i}
    H_i^{(u)} = G\left[\bigcup_{\{y\in \tp\;:\; |\H_i(y,u)|\geq (2k)^i\}}\H_i(y,u)\right].
\end{equation}
Note that, by construction, every edge $e\in \H_i(u)=\cup_{y\in \tp}\H_i(y,u)$ is either in $H_i^{(u)}$, or in the set $\F_i(y,u)$ where $y\in \tp$ is incident to~$e$.

\begin{figure}
    \centerline{\includegraphics[scale=0.5]{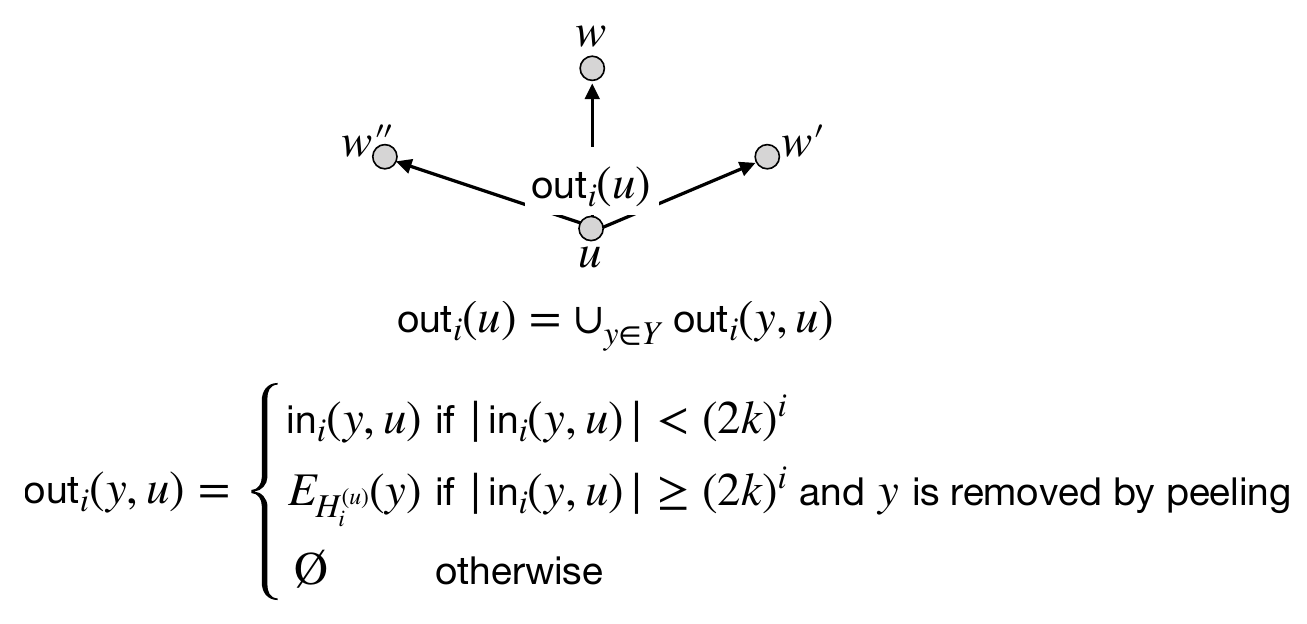}}
    \caption{\sl Construction of the set $\F_i(u)=\cup_{y\in Y}\F_i(y,u)$ from the sets $\H_i(y,u)$, $y\in Y$.}
    \label{fig:out-set}
\end{figure}

We are now going to update $H_i^{(u)}$ iteratively, by repeating the following sequence of ``peeling'' operations as long as they can be applied, i.e., as long as vertices can be removed. This sequence of operations bears similarities with the computation of the $k$-core of a graph, but the vertices of the partitions $\bp$ and $\tp$ of $H$ are here treated separately.  
\begin{center}
\fbox{
\begin{minipage}{15cm}
\centerline{\underline{The peeling process applied to $H_i^{(u)}$}}
\medskip
Repeat until no nodes can be removed: 
\begin{enumerate}
    \item\label{it:discard} Remove from $H_i^{(u)}$ all vertices $x\in \bp$ of degree smaller than~$k$, along with their incident edges in $H_i^{(u)}$, and update the degree of each vertex in $H_i^{(u)}$ accordingly; 
    \item\label{it:out-t} Remove from $H_i^{(u)}$ all vertices $y \in \tp$ of degree smaller than~$k$, along with their incident edges in $H_i^{(u)}$, and update the degree of each vertex in $H_i^{(u)}$ accordingly;
    \item For every $y\in \tp$, we set 
    \begin{equation}\label{eq:OUT_recur-bis}
     \mbox{$y$ removed at Instruction~2} \Longrightarrow  \F_i(y,u)=E_{H_i^{(u)}}(y)\mbox{ (i.e. the edges removed with $y$) };
    \end{equation} 
    That is, $\F_i(y,u)$ is defined as the set of edges incident to $y$ that were removed from $H_i^{(u)}$ together with~$y$ at this iteration. 
\end{enumerate}
\end{minipage}
} 
\end{center}
For any triple $i\geq 1$, $u\in V$, and $y\in \tp$ satisfying $|\H_i(y,u)|\geq (2k)^i$, if the value of $\F_i(y,u)$ has not been set in the above process, then it is set to $\F_i(y,u)=\varnothing$. That is, for every $y\in Y$, 
\[
\mbox{Eqs.~\eqref{eq:OUT_recur} and~\eqref{eq:OUT_recur-bis} do not apply}
\Longrightarrow \F_i(y,u)=\varnothing.
\]
Finally, we set
\begin{equation}\label{eq:OUT}
    \F_i(u) = \bigcup_{y \in \tp} \F_i(y,u).
\end{equation}
Note that, for every $i\geq 1$, $u\in V$, and $y\in \tp$, we have 
\[
\F_i(y,u)=E_H(y)\cap\F_i(u). 
\]
This equality is extended to define, for every triple $i\geq 1$, $u\in V$, and $x\in \bp$,
\[
\F_i(x,u)=E_H(x)\cap \F_i(u).
\]

\subsection{Core Graphs.} The graph resulting from the above iterated process of edge- and vertex-removal from $H_i^{(u)}$ is denoted by $\mathsf{Core}_i(u)$. 
The core graphs play an important role in our proof. Indeed, we will show (see Lemma~\ref{lem:cycle}) that if $\mathsf{Core}_i(u)$ is non-empty for some $i\geq 1$ and $u\in V$, then $G$ contains a $2k$-cycle. The density theorem will then follow from the fact (established in  Lemma~\ref{lem:ker}) that there exists $i$ and $u$ such that $\mathsf{Core}_i(u)$ is non-empty. 

\medbreak

Before proving Lemmas~\ref{lem:cycle} and~\ref{lem:ker}, let us establish a collection of statements for helping understanding the construction above. The fact below illustrates why one can view the sets $\F_i(u)$ and $\H_i(u)$ as produced by a (virtual)  protocol broadcasting edges of $H$ throughout the network~$G$. 

\begin{fact}\label{fact:OUT_to_IN}
    For every simple path $(u_0,\dots,u_\ell)$ in~$G$ with $u_0\in \bp$, we have that, for every $i\in\{1,\dots,\ell\}$, \[\F_{i-1}(u_{i-1})\subseteq\H_i(u_i).\]
\end{fact}

\begin{proof}
    By definition of set $\mathcal{Q}_i(u_i)$, the set of simple paths of length exactly~$i$ with one endpoint equal to~$u_i$, and the other endpoint equal to some node in~$\bp$, we have that, for every $1<i\leq \ell$,  \[(u_0,\dots,u_{i-1})\in\mathcal{Q}_{i-1}(u_{i-1}), \;\mbox{and}\; (u_0,\dots,u_{i-1},u_i)\in\mathcal{Q}_i(u_i).\]
    Then by definition (Eq.~\eqref{eq:IN_recur}) of $\H_i(u) = \bigcup\limits_{\{w\in N(u) \;\mid\; \exists P\in\mathcal{Q}_{i-1}(w),\; (P,u)\in\mathcal{Q}_i(u)\}} \F_{i-1}(w)$, it follows that $\F_{i-1}(u_{i-1})\subseteq\H_i(u_i)$.
\end{proof}

\begin{lemma}\label{lem:path_to_z}
    For every $i\in\{0,\dots,\ell\}$, $u\in V$, and $e=\{x,y\} \in \H_i(u)$ with $x\in \bp$ and $y\in \tp$, there is a simple path $(u_0,\dots,u_i)$ in $G$ such that $u_0=x$, $u_i=u$, and $e \in \cap_{j=0}^{i-1} \F_j(y,u_j)$.
\end{lemma}

To gain intuition about this statement, one can follow the idea of the virtual distributed protocol mentioned since the beginning of Section~\ref{subsec:set_in_out}. This lemma states that any edge $e=(x,y)$ that was received by $u$ during the virtual protocol was first broadcast by $x$ itself and forwarded along a path from $x$ to $u$.

\begin{proof}    
    Let $u_i=u$. We show the following statement by induction on~$j=i$ down to $j=1$. There is a simple path \[(P_{j-1},u_j,\dots,u_i)\in \mathcal{Q}_i(u_i)\] such that $P_{j-1}\in \mathcal{Q}_{j-1}(u_{j-1})$ for some node $u_{j-1}\in N(u_j)$, and $e\in \F_{j-1}(u_{j-1})$.
    
    \begin{itemize}
        \item The base case is $j=i$. By definition of $\H_i(u_i)$ (cf. Eq.~\eqref{eq:IN_recur}), there exists $(P_{i-1},u_i)\in \mathcal{Q}_i(u_i)$ such that $P_{i-1}\in \mathcal{Q}_{i-1}(u_{i-1})$ for some $u_{i-1}\in N(u_i)$, and $e\in\F_{i-1}(u_{i-1})$. Furthermore, by definition of $\mathcal{Q}_i(u_i)$, we have $u_i\notin P_{i-1}$.
    
        \item For the induction step, let us assume that the claim is true for $j+1$ where $j\in\{1,\dots,i-1\}$. By induction, $e\in \F_j(u_j)$, and, by construction of $\F_j(u_j)$, it also holds that $e\in \H_j(u_j)$. This follows directly from Eq.~\eqref{eq:OUT_recur}, or from Eq.~\eqref{eq:OUT_recur-bis} after having noticed that, thanks to Eq.~\eqref{eq:H_i}, $E(H_i^{(u)})\subseteq \H_i(u)$.         
        Using Eq.~\eqref{eq:IN_recur} as in the base case, all points in the claim are satisfied.
    \end{itemize}
    Applying the claim for $j=1$, we get that there exists a path $(P_0,u_1,\dots,u_i)$ such that $P_0\in \mathcal{Q}_0(u_0)$ for some node $u_0\in N(u_1)$, $e\in \bigcap_{j=0}^{i-1}\F_j(u_j)$, and $u_j\notin \{u_0,\dots,u_{j-1}\}$ for all $j\in\{1,\dots,i\}$.
    The latter ensures that the path $(u_0,\dots,u_i)$ is simple. Using the definition of the set $\mathcal{Q}_j$ for $j=0$, it follows that $P_0=(u_0)=(x)$. 
    Moreover, since $e$ is incident to $y$, and since $e\in \bigcap_{j=0}^{i-1}\F_j(u_j)$, we get    
    $e\in \bigcap_{j=0}^{i-1}\F_j(y,u_j)$.
\end{proof}

\begin{fact}\label{fact:deg_F}
    For all $i\in\{1,\dots,\ell\}$ and $u\in V$, every node $y \in \tp \cap \mathsf{Core}_i(u)$ satisfies $|\H_i(y,u)|\geq (2k)^i$.
\end{fact}

\begin{proof}
    The claim follows directly from Eq.~\eqref{eq:H_i}, and the fact that $\mathsf{Core}_i(u)\subseteq H^{(u)}_i$.
\end{proof}

\begin{fact}\label{fact:deg_ker}
    For all $i\in\{1,\dots,\ell\}$ and $u\in V$, every node $w\in\mathsf{Core}_i(u)$ is of degree at least $k$ in the graph $\mathsf{Core}_i(u)$. 
\end{fact}

\begin{proof}
    The claim follows directly from the fact that, in the construction of  $\mathsf{Core}_i(u)$, all nodes $w\in \bp\cup\tp$ with degree smaller than~$k$ are removed (cf. Steps~\ref{it:discard} and~\ref{it:out-t} in the peeling).
\end{proof}

\begin{fact}\label{fact:OUT_size}
    For every $i\in\{1,\dots,\ell\}$, $y\in \tp$, and $u\in V$, we have $|\F_i(y,u)|< (2k)^i$.
\end{fact}

\begin{proof}
    The claim follows from the construction of set $\F_i(y,u)$. If the set was constructed by Eq.~\eqref{eq:OUT_recur} then $\F_i(y,u)=\H_i(y,u)$, and its size is smaller than $(2k)^i$. If the set $\F_i(y,u)$ was constructed by Eq.~\eqref{eq:OUT_recur-bis} then it contains edges adjacent to $y$ in $H_i^{(u)}$ that were removed, precisely because the current degree of~$y$ was smaller than~$k$. The remaining case is $\F_i(y,u)=\varnothing$ for which the claim holds trivially.  
\end{proof}

\begin{fact}\label{fact:size_IN}
    For every $i\in\{1,\dots,\ell\}$, $u\in V$, and  $x\in \bp$, we have \[|\H_i(x,u)|\leq|\F_i(x,u)|+\deg_{\mathsf{Core}_i(u)}(x)+k-1.\]    
\end{fact}

\begin{proof}
Let us consider an edge $e=\{x,y\}\in \H_i(x,u)$ with $x\in \bp$ and $y\in \tp$. By Eq.~\eqref{eq:inx-equal-iny}, we have $e\in \H_i(y,u)$. The edge $e$~satisfies one and only one of the following four cases. 
    \begin{enumerate}
        \item $e$ joins $\F_i(y,u)$ by applying Eq.~\eqref{eq:OUT_recur};    
        \item $e$ is removed from $H_i^{(u)}$ along with vertex $x\in \bp$ (cf. Step~\ref{it:discard} in the peeling); 
        \item $e$ is removed from $H_i^{(u)}$ along with vertex $y\in \tp$ (cf. Step~\ref{it:out-t} in the peeling) --- the edge $e$ is then added to $\F_i(y,u)$ according to Eq.~\eqref{eq:OUT_recur-bis};
        \item $e$ is never removed from $H_i^{(u)}$ --- in this case, $e$ belongs to $\mathsf{Core}_i(u)$; 
    \end{enumerate}
     In the first and third cases, we have $e\in \F_i(u)$ thanks to Eq.~\eqref{eq:OUT}, and thus $e\in \F_i(x,u)$. The second case applies to at most $k-1$ edges from $\H_i(x,u)$. Finally, at most $\deg_{\mathsf{Core}_i(u)}(x)$ edges satisfy the fourth case.
\end{proof}

We are now ready to establish one of the two main arguments in the proof of our density theorem. 

\begin{lemma}\label{lem:cycle}
    If there exist $i\in \{1,\dots,\ell\}$ and $u\in V$ such that $\mathsf{Core}_i(u)\neq\varnothing$ then there is a $2k$-cycle in~$G$.
\end{lemma}
\begin{proof}    
We construct the aforementioned paths $P,P'$, and $P''$ whose union forms a $2k$-cycle (see Figure~\ref{fig:paths-PPP}). Note that, as a subgraph of $\H_i(u)$, which is itself a subgraph of the bipartite graph~$H$, $\mathsf{Core}_i(u)$ is also bipartite. Its two parts are merely $\bp\cap V(\mathsf{Core}_i(u))$ and $\tp\cap V(\mathsf{Core}_i(u))$. Let $e=\{x,y\}$ be an edge of $\mathsf{Core}_i(u)$ such that $x\in \bp$ and $y\in \tp$. By Lemma~\ref{lem:path_to_z}, there exists a simple path \[P=(u_0,\dots,u_i)\] such that $u_0=x$, $u_i=u$, and
\[
\{x,y\}\in \bigcap_{j=0}^{i-1} \F_j(y,u_j).
\]

\begin{figure}[htb]
    \centerline{\includegraphics[scale=0.5]{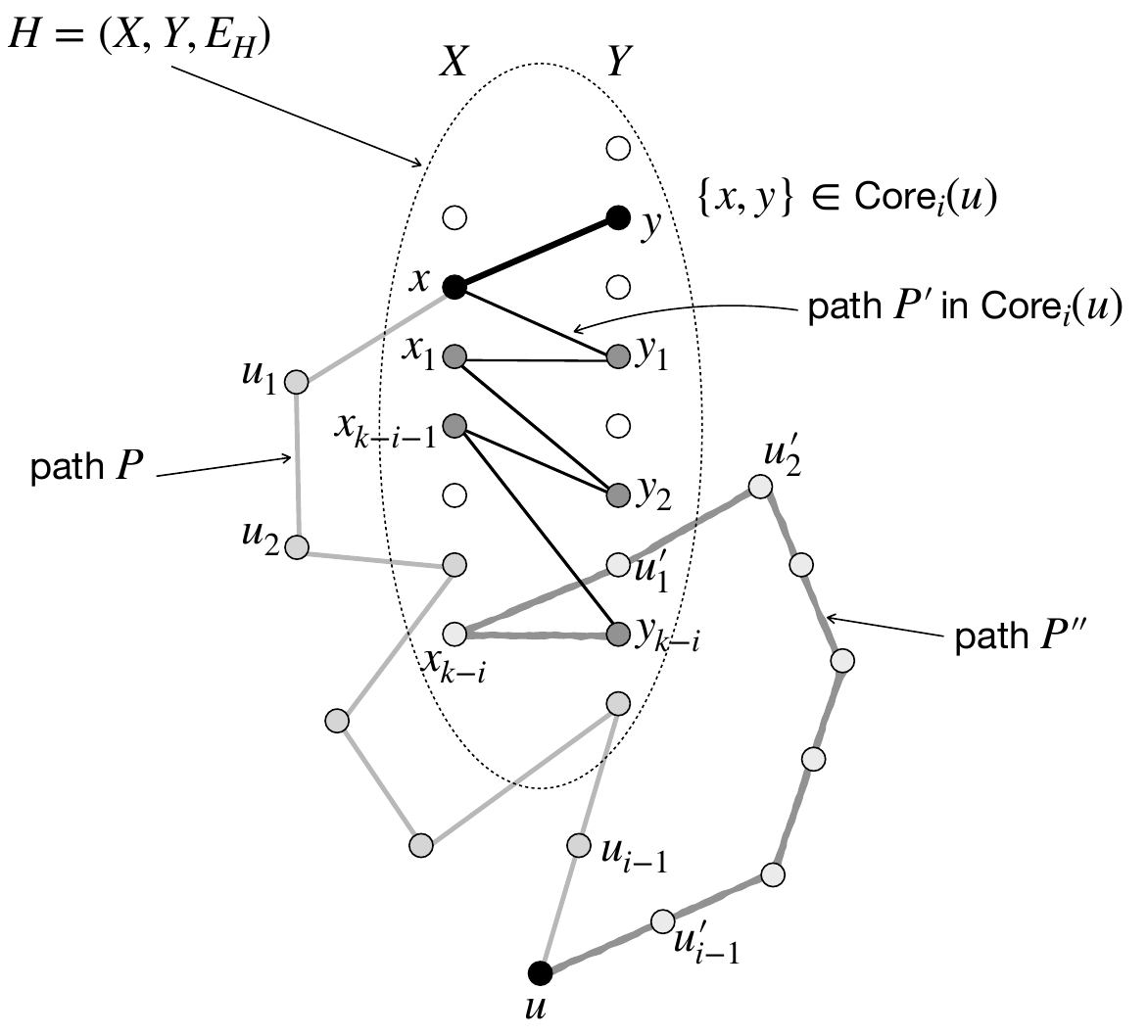}}
    \caption{\sl Construction of the paths $P,P'$, and $P''$ in the proof of Lemma~\ref{lem:cycle}}
    \label{fig:paths-PPP}
\end{figure}

We aim at constructing a path $P'$ in $\mathsf{Core}_i(u)$ that starts at $x_0=x$, and ends at some node $y_{k-i}$, of length $2(k-i)-1$. We build this path iteratively, by increasing its length. 
For the base case, note that $x\in \bp\cap\mathsf{Core}_i(u)$, and thus, by Fact~\ref{fact:deg_ker}, we get that \[\deg_{\mathsf{Core}_i(u)}(x)\geq k>\ell\geq i.\] It follows that there exists 
\[y_1\in N_{\mathsf{Core}_i(u)}(x)\smallsetminus\{u_1,\dots,u_i\}.\] 
The node $y_1$ belongs to $\tp$, and the path $P'$ is initialized to $P'_1=(x_0,y_1)=(x,y_1)$.
For the induction step, let us assume that a path 
\[P'_j=(x_0,y_1,x_1,\dots,y_{j-1},x_{j-1},y_j)\] 
has been constructed, with $x_0=x$ and $1\leq j<k-i$. As $y_j\in \tp\cap \mathsf{Core}_i(u)$, and since, by Fact~\ref{fact:deg_ker},
\[
\deg_{\mathsf{Core}_i(u)}(y_j)\geq k>i+j,
\]
we get that there exists 
\[
x_j\in N_{\mathsf{Core}_i(u)}(y_j)\smallsetminus\{u_1,\dots,u_i,x_0,x_1,\dots,x_{j-1}\}.
\]
We have $x_j\in \bp$. The path $P'_j$ can thus be extended into the path $Q'_j=(x_0,y_1,x_1\dots,y_j,x_j)$. Moreover, as $x_j\in \bp\cap\mathsf{Core}_i(u)$ and 
\[\deg_{\mathsf{Core}_i(u)}(x_j)\geq k>i+j,\]
we get that there exists 
\[y_{j+1}\in N_{\mathsf{Core}_i(u)}(x_j)\smallsetminus\{u_1,\dots,u_i,y_1,\dots,y_j\}.\] 
We have $y_{j+1}\in \tp$. The path $Q'_j$ can thus be extended into the path 
\[P'_{j+1}=(x_0,y_1,x_1,\dots,y_j,x_j,y_{j+1}).\]
We can proceed with the construction until we get a path 
\[P'=P'_{k-i}=(x_0,y_1,x_1,\dots,y_{k-i-1},x_{k-i-1},y_{k-i})\]
of length $2(k-i)-1$, as desired.

We now aim at extending $P\cup P'$ into a $2k$-cycle, by constructing a simple path $P''$ starting at $y_{k-i}$ and ending at $u$ that does not intersect $P\cup P'$. For this purpose, let us consider the set of edges: 
\[
A=\bigcup_{w\in P\cup P'}\bigcup_{j=1}^{i-1} \F_j(y_{k-i},w).
\]
Together, $P$ and  $P'$ contain exactly $2k-i$ vertices. By Fact~\ref{fact:OUT_size}, we have that, for every $j\in\{1,\dots,i-1\}$, \[|\F_j(y_{k-i},w)| \leq (2k)^j,\] 
from which it follows that
\[
|\bigcup_{w\in P\cup P'}\F_j(y_{k-i},w)|\leq (2k-i)(2k)^j.
\]
Therefore, 
\[|A| \leq (2k-i)\cdot \sum_{j=1}^{i-1}(2k)^j.\]
Let us now consider  the set $\bp_{\text{bad}}$ of vertices in $P\cup P'$ that are also in $\bp$, i.e., 
\[
\bp_{\text{bad}}= \bigl(\bp\cap\{u_1,\dots,u_i\}\bigr)\cup\{x_0,\dots,x_{k-i-1}\}.
\]
We have $|\bp_{\text{bad}}| \leq k$, and thus
\begin{align*}
|A|+|\bp_{\text{bad}}|
& \leq (2k-i) \cdot \sum_{j=1}^{i-1}(2k)^j + k \\
& = (2k-i) \cdot \sum_{j=0}^{i-1}(2k)^j - (2k-i) + k \\
& = (2k-i) \cdot \frac{(2k)^i-1}{2k-1} - k+i \\
& \leq ((2k)^i-1) -k +i \\
& < (2k)^i.
\end{align*}
Since $y_{k-i}\in \mathsf{Core}_i(u)$, it follows from Fact~\ref{fact:deg_F} that $|\H_i(y_{k-i},u)|\geq (2k)^i$. As a consequence, 
\[|A|+|\bp_{\text{bad}}|<|\H_i(y_{k-i},u)|.\]
Therefore, there exists $x_{k-i}\in N_{\H_i(u)}(y_{k-i})$ such that $x_{k-i}\notin P\cup P'$, and, for every $j\in\{1,\dots,i-1\}$ and every $w\in P\cup P'$,
\[\{y_{k-i},x_{k-i}\}\notin \F_j(w).\]
By Lemma~\ref{lem:path_to_z}, there exists a path $P''=(u'_0,\dots,u'_i)$ such that $u'_0=x_{k-i}$, $u'_i=u$, and, for every $j\in\{1,\dots,i-1\}$, 
\[\{y_{k-i},x_{k-i}\}\in\F_j(u'_j).\] 
Thus $P''$ does not intersect $P'$, and it  intersects $P$ only at~$u=u_i$. Therefore, the union of the three paths 
\[
P\cup P'\cup P''= (u,u_{i-1},\dots,u_1,x_0,y_1,x_1,\dots,x_{k-i-1},y_{k-i},x_{k-i},u'_1,\dots,u'_{i-1})\]
forms a $2k$-cycle, which concludes the proof. 
\end{proof}

\begin{lemma}\label{lem:ker}
    There exists $i\in\{1,\dots,\ell-1\}$ and $u \in V$ such that $\mathsf{Core}_i(u)\neq\varnothing$.
\end{lemma}

\begin{proof}
    The proof goes by contradiction. Let us assume that, for every $i\in\{1,\dots,\ell-1\}$ and every $u \in V$, we have $\mathsf{Core}_i(u) = \varnothing$. We are going to show that this implies $|E_H|< \tau n/6$, contradicting Fact~\ref{fact:size_edges} (recall that $\tau$ was defined in Equation~\ref{eq:tau}).
    Let $x\in \bp\subseteq R(v)$. By definition of $R(v)$, there exists a simple path $(u_0,\dots,u_\ell)$ of length $\ell$ between $u_\ell=v$ and~$u_0=x$ in~$G$. By definition, $\mathcal{Q}_0(x)=(x)$, and $\F_0(x,x)=E_H(x)$. For establishing the contradiction, let us revisit the recursive construction of the sets $\F_i(x,u_i)$ for $i=1$ to~$\ell$. 
    By Fact~\ref{fact:OUT_to_IN}, for all $i\in\{1,\dots,\ell\}$, we have $|\H_i(x,u_i)| \geq |\F_{i-1}(x,u_{i-1})|$. Thus, since 
     $\mathsf{Core}_i(u)=\varnothing$, Fact~\ref{fact:size_IN} yields that
    \[|\F_i(x,u_i)|\geq |\H_i(x,u_i)|-(k-1)\geq |\F_{i-1}(x,u_{i-1})|-(k-1).\]
    Therefore, we get that, for every $x\in \bp$ and $i\in \{1,\dots,\ell\}$, 
    \[|\F_i(x,u_i)| \geq |\F_0(x,x)|-i(k-1) = |E_H(x)|-i(k-1).\] 
    As $\bp$ is one of the two parts of the bipartite graph $H$, Fact~\ref{fact:OUT_size} implies  that
    \begin{align*}
        |E_H|
        &=\sum_{x\in \bp}|E_H(x)|\\
        &\leq \sum_{x\in \bp}\Big(|\F_\ell(x,u)|+\ell(k-1)\Big)\\
        &= |\F_\ell(u)|+\sum_{x\in \bp}\ell(k-1)\\
        &=\sum_{y\in \tp}|\F_\ell(y,u)|+\ell(k-1)\cdot |\bp|\\
        &< (2k)^\ell\cdot|\tp|+\ell(k-1)\cdot|\bp|\\
        &<(2k)^\ell\cdot(|\tp|+|\bp|)\\
        &\leq (2k)^\ell\cdot n\\
        &= \frac{\tau}{6}\cdot n
    \end{align*}
The latter inequality is the desired contradiction.  
\end{proof}

\begin{proof}[Proof of Theorem~\ref{lem:density}]
    By Lemma~\ref{lem:ker}, there exists $i\in\{1,\dots,\ell-1\}$ and $u\in V$ such that $\mathsf{Core}_i(u)\neq\varnothing$. The existence of a $2k$-cycle can then be concluded with Lemma~\ref{lem:cycle}.
\end{proof}

\section{Conclusion}\label{sec:conclu}

We have proved that, for every $k\geq 2$, there exists a \emph{deterministic} distributed algorithm that decides whether the input graph~$G$ is $C_{2k}$-free in $O(n^{1-\nicefrac{1}{k}})$ rounds under the \CONGEST\/ model. This result is based on a new result in graph theory, which essentially states that when some form of  ``local density'' exceeds a certain threshold (that depends on~$k$) in a graph, that graph must contain a $2k$-cycle. 

We point out that our deterministic algorithm for $C_{2k}$-freeness can be used to solve the same problem in the Quantum $\CONGEST$ model in $\tilde{O}(n^{\nicefrac{1}{2}-\nicefrac{1}{2k}})$ rounds, following the same approach as in~\cite{FraigniaudLMT24}. Informally, the approach consists in three steps : (1) apply the “diameter reduction” technique of~\cite{eden2022sublinear}, which allows to reduce the problem to graphs of polylogarithmic diameter, (2) operate a “congestion reduction” on Algorithm~\ref{alg_C2k} to make it work in a constant number of $\CONGEST$ rounds, at the cost of reducing the probability for cycle detection to $\Theta(n^{-1+\nicefrac{1}{k}})$, and (3) eventually use an “amplification technique” based on quantum Grover search to obtain a constant probability of detecting a $C_{2k}$, if exists, in $\tilde{O}(n^{\nicefrac{1}{2}-\nicefrac{1}{2k}})$ rounds. The same round complexity was already attained in~\cite{FraigniaudLMT24}, but the new algorithm is considerably simpler.

Note that the constant $\lebontau$ in our density theorem is not tight, at least for some values of $k$ and $\ell$. For instance, for $k=2$, and $\ell=1$, our theorem states that if there exists $v\in V$ such that $\sum_{u\in N(v)}\deg(u) > 24\cdot n$, then there is a $4$-cycle in $G$ whereas, in fact, there is a $4$-cycle in $G$ already when $\sum_{u\in N(v)}\deg(u) > 2\cdot n$. We let as an open problem the determination, for every $k\geq 2$ and $\ell\in\{1,\dots,k-1\}$, of the smallest value $\tau=\tau(k,\ell)$ such that the existence of a node~$v$ for which $\sum_{u\in R_\ell(v)}\deg(u) > \tau \cdot n$ implies the existence of a $2k$-cycle. 

Arguably one of the main open problems in the field of distributed subgraph detection under the \CONGEST\/ model is whether $O(n^{1-\nicefrac{1}{k}})$ is the best round-complexity that can be achieved for $2k$-cycle detection, whether it be by deterministic or randomized algorithms, up to polylogarithmic factors. It is known to be the case for $k=2$, i.e., the round-complexity of $C_4$-freeness is $\tilde\Omega(\sqrt{n})$, but it is open for $k>2$. In other words, is it true that, for every $k\geq 2$, $C_{2k}$-freeness cannot be decided under the \CONGEST\/ model in $\tilde o(n^{1-1/k})$ rounds?

\bibliography{biblio-cycle}
\end{document}